\documentclass[10pt, conference, letterpaper]{IEEEtran}

\usepackage{times}
\usepackage{stackrel}

\usepackage{epsf}
\usepackage{amsmath,amssymb,amsthm}
\usepackage{graphicx}
\usepackage{subfigure}
\usepackage{color}
\usepackage{cite}
\usepackage{multirow}
\usepackage{ifthen}
\usepackage{tabularx}
\usepackage{epstopdf}
\usepackage{bbm}
\usepackage{bm}
\usepackage{algorithm}
\usepackage{algpseudocode}

\newtheorem{theorem}{{\bf Theorem}}
\newtheorem{lemma}{{\bf Lemma}}

\newtheorem{proposition}{\bf Proposition}

\newcommand{\brac}[1]{\left({#1}\right)}
\newcommand{\sbrac}[1]{\left[{#1}\right]}
\newcommand{\cbrac}[1]{\left\{{#1}\right\}}

\newcommand{\mc}[1]{\mathcal{#1}}
\newcommand{\mb}[1]{\mathbb{#1}}

\newcommand{\abs}[1]{\lvert{#1}\rvert}

\newcommand{\norm}[1]{\|{#1}\|}

\newcommand{\indic}{\mathbbm{1}}
\newsavebox{\mybox}\newsavebox{\mysim}
\newcommand{\distras}[1]{%
  \savebox{\mybox}{\hbox{\kern3pt$\scriptstyle#1$\kern3pt}}%
  \savebox{\mysim}{\hbox{$\sim$}}%
  \mathbin{\overset{#1}{\kern\z@\resizebox{\wd\mybox}{\ht\mysim}{$\sim$}}}%
  \newcommand{\fact}[1]{{#1}\!}
}

\begin{document}

\title{Optimal Content Replication and Request Matching in Large Caching Systems}

%

%

\author{\IEEEauthorblockN{Arpan Mukhopadhyay\IEEEauthorrefmark{1}
, Nidhi Hegde\IEEEauthorrefmark{2} and Marc Lelarge\IEEEauthorrefmark{3}}
\IEEEauthorblockA{\IEEEauthorrefmark{1}IC/LCA-2, EPFL, Switzerland. Email: arpan.mukhopadhyay@epfl.ch\\}
\IEEEauthorblockA{\IEEEauthorrefmark{2}Nokia Bell Labs, France. Email: nidhi.hegde@nokia-bell-labs.com\\}
\IEEEauthorblockA{\IEEEauthorrefmark{3}INRIA-ENS, Paris, France. Email: marc.lelarge@ens.fr}}

\maketitle

\begin{abstract}
We consider models of content delivery networks in which the servers are constrained by 
two main resources: memory and bandwidth. In such systems, 
the throughput crucially depends on how contents are replicated 
across servers and  how the requests of specific contents 
are matched to servers storing those contents.  
In this paper, we first formulate the  problem of 
computing the optimal replication policy which if combined with the optimal 
matching policy maximizes the throughput of the caching system in the stationary regime.  
It is shown that computing 
the optimal replication policy for a given system is an NP-hard problem.  
A greedy replication scheme is proposed and it is shown that the scheme 
provides a constant factor approximation guarantee.  
We then propose a simple randomized matching scheme 
which avoids the problem of interruption in service of the ongoing requests  due to
re-assignment or repacking of the existing requests 
in the optimal matching policy. 
The dynamics of the caching system is
analyzed under the combination of proposed replication and matching schemes.   
We study a limiting regime, where the number of servers and 
the arrival rates of the contents are scaled proportionally, 
and show that the proposed policies achieve asymptotic optimality.  
Extensive simulation results are presented to evaluate the performance of different policies
and study the behavior of the caching system under different service time distributions
of the requests.
\end{abstract}

%
%
%

%
%

%
%



\section{Introduction}

Recent years have seen an explosive growth
in Internet traffic, stemming mainly from the transfer 
of  multi-media contents, e.g.,  streaming videos, movies etc.  
It is expected that video streaming services and downloads will account for more than
81\% of  all Internet's traffic by 2021~\cite{CiscoWP17}.
Such growth in multi-media traffic  has led to the emergence of
content delivery networks (CDNs) and peer-to-peer systems, 
which support the demand for contents by replicating
popular contents at the network periphery (e.g., boxes or servers).
Popular video-streaming services such as Netflix, Youtube
often use CDN's to serve the requests of their most popular contents.     

Large CDNs usually consist of a central server, storing an entire catalogue of contents,
and a large number of edge servers, each storing a small fraction of these contents
in their caches and serving requests of the stored contents~\cite{Dilley2002}.
In such systems, it is assumed that access to the central server  is expensive. 
Therefore, a large portion of the content requests must be served by the edge servers that 
are constrained by their limited memory and bandwidth capacities. 
In this paper, we model these servers as loss servers~\cite{Kelly_loss} and aim at 
minimizing the number of requests 
blocked at these servers (and thus need to be sent to the central server).  
We do not consider queueing of the requests since we focus on delay-sensitive streaming services, 
which comprise a large proportion of the Internet's traffic today~\cite{CiscoWP17}.

Efficiency of such systems crucially depends on
the replication (also called allocation) policy used to populate the caches of the servers and
the request matching policy used to dispatch the incoming requests. 
In this paper, 
we first formulate the problem of computing the optimal
allocation policy, which, if combined with the optimal 
matching policy (maximum matching), maximizes the number requests served per unit time by the caching system
in the stationary regime. 
To the best of our knowledge, this is the first work that addresses
this joint allocation-matching problem.   
%
The joint problem for finite systems
is shown to be an  NP-hard problem by reduction
from the 3-partition problem.   
A polynomial-time greedy allocation scheme is proposed
and it is shown to achieve a constant factor approximation
of the optimal value. 

Next, we turn our attention to the dynamics of the system
which are determined by the matching policy in use. 
In earlier works, e.g., \cite{Tan_massoulie_Caching, Lelarge_caching}
a maximum matching scheme has been considered to match incoming requests to servers.
However, in the maximum matching scheme, the service of ongoing requests may be interrupted 
to incorporate a newly arrived request. Such interruption in service is clearly 
not desired and may cause significant delay in serving the requests. 
 We thus propose a simple randomized policy for request matching 
which does not cause interruption to the already existing requests.  We show that
this matching scheme coupled with the proposed greedy replication policy
or a previously studied~\cite{Tan_massoulie_Caching, Kleinrock_caching}
`proportional-to-product' replication policy
is asymptotically optimal in a limiting regime where the number of
cache servers and request arrival rates scale proportionally with each
other and the number of contents remains fixed.   Such a scaling regime corresponds to 
scenarios where a fixed number of most popular contents are served by a large number of cache servers.  
The proof of asymptotic optimality uses fluid limits of processes describing the dynamics of the system.
The fluid limit in our case cannot be described 
by ordinary differential equations (ODE's) since it
describes a non-smooth dynamical system.
The novelty in our approach lies in describing the fluid limit as a solution to a differential
inclusion (DI) system and showing that all trajectories of the solutions
converge to the same global attractor. In summary, our main contributions are as follows.

\begin{itemize}
\item We show that the joint problem is NP-hard and propose a
  polynomial-time approximation algorithm that achieves performance
  within a constant 
  of the optimal value;
  
\item We propose a randomized matching policy that avoids interruption of service 
and show that this policy is asymptotically optimal for large systems
when combined with the proposed allocation policies.

\item We employ a new approach based on the theory of differential inclusions 
to prove fluid limit results.

\item We also present extensive simulation results to compare the different 
allocation and matching schemes discussed in this paper. Near insensitivity of
the system to service time distributions is also studied. 
\end{itemize}

\noindent\paragraph*{Related Work}  Content placement in caching systems has been the 
subject of study for many years now.  
Of the numerous papers on this topic, we now mention a few that are most relevant to our work. 
In \cite{wireless_caching, Kleinrock_caching, Anand_caching}, optimal cache allocation policy 
was designed without taking into account the bandwidth restrictions of the servers.
A loss model of caching systems was first introduced in~\cite{Tan_massoulie_Caching}
in the context of peer-to-peer video on demand services. 
A replication policy in which the number of replicas is proportional
to the arrival rate of the contents was proposed and analyzed in conjunction
with the maximum matching algorithm which is different from the setting in our paper. 
In~\cite{Lelarge_caching}
a similar loss model was considered. However, the objective was to 
maximize the utilization of the resources as opposed to maximization of
throughput. 
In~\cite{Moharir_caching}, a discrete-time model
similar to ours is considered. The objective there is to minimize the 
expected transmission rate from the main server
in order to serve all requests in a time slot. 
A different scaling regime in which the number of contents is scaled is 
considered. We do not consider such a scaling regime in this paper 
since our focus is on a scenarios where only a fixed number of highly popular contents are present at any instant.
An online matching policy similar to the proposed matching policy was considered in
\cite{Stolyar_loss} for cloud computing systems. However, the setting there is completely different
as the servers do not have any memory restrictions.

The remainder of the paper is organized as follows.
Section~\ref{sec:model} introduces the system model.  In Section~\ref{sec:form} we formulate
the joint allocation and matching problem and analyze  the complexity of the problem.  In Section~\ref{sec:algos} we present  efficient approximate algorithms.
In Section~\ref{sec:grand} we present a matching policy that does not involve repacking of ongoing requests, and analyze the system in the large-systems asymptotic regime.
 We prove optimality in the
large-systems scaling regime. Section~\ref{sec:numerics} presents 
simulation results.
Finally, the paper is concluded in
Section~\ref{sec:conclusion}.

\section{System model}
\label{sec:model}

We consider a dynamic 
model of caching systems in which
requests for contents arrive at random instants and are served
by corresponding servers.  
The servers are assumed to be able to serve 
only a finite number of requests simultaneously.

\subsection{Server and storage Model} 
\label{sec:storage}

The caching system consists of $n$ servers and $m$ contents 
indexed by the sets $S=\cbrac{1,2,\ldots,n}$ and $C=\cbrac{1,2,\ldots,m}$, respectively.
We assume that each server $s \in S$ is capable
of storing up to $d_s \geq 1$ contents in its cache
and has a bandwidth of $U_s >0$, i.e., it can serve $U_s$ requests simultaneously. 
The {\em replication}, or \emph{allocation} policy is represented by a binary matrix $A=(a_{sc})_{s \in S, c \in C} \in \cbrac{0,1}^{nm}$,
with $a_{sc}=1$ if $c$ is stored in the cache of $s$ and $a_{sc}=0$, otherwise.
Thus, for any {\em feasible replication policy} $A=(a_{sc})_{s \in S, c \in C} $,
we have

\begin{equation}
\sum_{c \in C} a_{sc} \leq d_s \text{ for all } s \in S
\label{eq:deg_cons}
\end{equation}
The set of all feasible replication policies is denoted as $\mathcal{A} \subset \cbrac{0,1}^{nm}$.
The cache of each server is populated at $t=0$
according to some cache allocation policy $A \in \mc{A}$ 
and is kept unchanged for all $t \geq 0$. 
Servers are assumed to have 
different bandwidths and memory sizes from the sets $\cbrac{U_i, i \in \mc{I}}$
and $\cbrac{d_j, j \in \mc{J}}$, respectively, where $\mc{I}$ and $\mc{J}$ are some finite index sets.
The fraction of servers having bandwidth $U_i$ and memory size
$d_j$ is denoted as $\alpha_{ij}$ for each $(i,j) \in \mc{I} \times \mc{J}$.

\subsection{Service model}

Requests of contents are assumed to arrive one at a time according to simple point processes.
The average number of requests of a content $c \in C$
arriving per unit time is denoted by $\lambda_c$ and is called the {\em arrival rate}
or the {\em popularity} of the content.
The vector of content popularities is denoted as $\Lambda=(\lambda_c, c \in C)$ and is
assumed to be a known constant throughout the paper.
Typically, the popularities of news items containing videos or movies vary over periods
ranging from few hours to few weeks~\cite{femto_caching} (and have to be estimated periodically~\cite{Moharir_sigm_2014})
which are slower than the time scales of interest in the paper.
For this reason we do not consider the effect of time-varying popularities and errors in 
estimation of the popularities in this paper. 

Arriving requests are matched or assigned to servers according to some
matching scheme. Clearly, no matching policy can serve all arriving requests since
the servers have limited storage and bandwidth capacities.
The requests which cannot be served are {\em blocked} or dropped by the caching system and 
are immediately routed to a central server 
that stores all the contents.  
The accepted requests are assumed to stay in the system
for a random amount of time, independent and identically distributed (i.i.d) with unit mean. 

\section{Optimal allocation: A static formulation}
\label{sec:form}

Our first objective is the design of an allocation policy
which populates the caches at the start of the system.  
We seek an allocation $A$ that solves the {\em static} optimization problem of maximizing
the average number requests served per unit time by the caching system
operating in the stationary regime.
To formulate the problem, we consider a time window of unit length.
Let $X_c$, $c \in C$ denote the number of requests of content $c$
arriving in this window and define $X:=(X_c, c \in C)$.  We make the following approximations: 1) the cache servers are idle at the beginning of the time window, 2) all the requests that 
are accepted by the system in this time window stay in the system exactly till the end of the time window.
Although these approximations do not capture dynamics of the departures of the requests in this window,
they provide a good approximation of the stationary behavior of the caching system 
(since each request has an average service time of unit duration)
and are appropriate  when detailed statistics
of the arrival processes and service times are not available. We analyze the dynamics
of the caching system under specific assumptions on arrival processes and service time distributions
later in this paper.

Let $b_{sc}$ denote the total number
of requests of content $c \in C$ matched/assigned to a server $s \in S$ in the given time window.
Then, $B=(b_{sc})_{s \in S c \in C}$ must satisfy 
\begin{align}
&  \sum_{c \in C} b_{sc}  \leq U_s, \text{ } \forall s \in S, \label{eq:cons_u}\\
& \sum_{s \in S} b_{sc} \leq X_c, \text{ } \forall c \in C, \label{eq:cons_x}\\
& \indic\brac{b_{sc} > 0} \leq a_{sc}, \label{eq:cons_a}
\end{align}
where $\indic\brac{\cdot}$ denotes the indicator function.
Let $\mc{B}(A,X)$ denote the set of matchings $B$ which satisfy \eqref{eq:cons_u}-\eqref{eq:cons_a}.
Thus, under allocation policy $A$, the maximum total number of requests
that can be matched in this time window is $M(A,X)=\max_{B \in \mc{B}(A,X)}\sum_{s \in S c \in C} b_{sc}$.
Our goal is to find an allocation policy $A$ for which $M(A,X)$ is maximized, 
i.e., we aim to find a solution to $\max_{A \in \mc{A}} M(A,X)=\max_{A \in \mc{A}}\max_{B \in \mc{B}(A,X)} \sum_{s \in S, c \in C} b_{sc}$.


We note that in the above formulation the
random vector $X$ appears in the constraints. 
Hence, the solution $A^*$ of the above problem will only be optimal
in windows where the demand vector is exactly equal to $X$.
However, since random demands are unknown a priori and our goal is to find a solution which works well on average, 
we replace the vector
$X$ by its mean value, i.e., the {\em popularity vector}
$\Lambda=(\lambda_{c}, c\in C)$. Hence, we  consider the following problem:

\begin{equation}
\max_{A \in \mc{A}} M(A,\Lambda)=\max_{A \in \mc{A}}\max_{B \in \mc{B}(A,\Lambda)} \sum_{s \in S, c \in C} b_{sc}
\label{opt_semifinal}
\end{equation}
We refer to the above as the joint allocation-matching (JAM) problem
since in it the decision variable $A$ also affects the matching $B \in \mc{B}(A,\Lambda)$.
Our first result establishes the following equivalence.

\begin{proposition}
Problem~\eqref{opt_semifinal} is equivalent to the following problem
\begin{equation}
 \begin{aligned}
 & \underset{Z=(z_{sc})_{s\in S, c \in C}}{\text{Maximize}} & & \sum_{s \in S} \sum_{c \in C} z_{sc} \\
 & \text{subject to} & & \sum_{c \in C} z_{sc} \leq U_s, \text{ } \forall s \in S\\
 &&&  \sum_{s \in S} z_{sc}  \leq \lambda_c, \text{ } \forall c \in C\\
 &&&  \sum_{c \in C} \indic \brac{z_{sc} > 0}  \leq d_s, \text{ } \forall s \in S\\
 &&&  z_{sc} \in \mb{R}_+, \text{ } \forall s \in S, \forall c \in C 
 \label{opt_final}
 \end{aligned}
\end{equation}
Furthermore, if $Z^*=({z}_{sc}^*)_{s \in S, c \in C}$ 
is an optimal solution of \eqref{opt_final}, then 
an optimal solution ${A}^*$ of 
problem~\eqref{opt_semifinal}
can be found by setting ${a}_{sc}^*=\indic\brac{{z}_{sc}^* > 0}$, for all $(s,c) \in S \times C$. 
\end{proposition}

\begin{proof}
Let ${Z}=({z}_{sc})_{s \in S, c \in C}$ be a feasible solution
of~\eqref{opt_final}. 
Set ${a}_{sc}=\indic\brac{{z}_{sc} > 0}$ and ${b}_{sc}={z}_{sc}$
for all $s \in S$, $c \in C$.
Then clearly we have ${B}=({b}_{sc})_{s \in S, c \in C} \in \mc{B}\brac{A,\lambda}$. Furthermore,
$\sum_{s \in S} \sum_{c \in C} {z}_{sc}= \sum_{s \in S} \sum_{c \in C}{b}_{sc} \leq O_1$,
%
where $O_1$ denotes the optimal value of problem~\eqref{opt_semifinal}. Taking the maximum of
the LHS over the set of feasible solutions of \eqref{opt_final} (this is possible since
the feasible set of solutions is compact and the objective function is continuous) we have
$O_2 \leq O_1$,
%
where $O_2$ denotes the optimal value of problem~\eqref{opt_final}.

Conversely, suppose that ${A}=({a}_{sc})_{s \in S , c \in C} \in \mc{A}$,
${B}=({b}_{sc})_{s \in S , c \in C} \in \mc{B}(A,\lambda)$. 
For all pairs
$(s, c) \in S\times C$ 
such that ${a}_{sc}=0$ we set ${z}_{sc}=0$. 
For all other pairs
$(s,c) \in S \times C$ we set ${z}_{sc}={b}_{sc}$.
Clearly, $(z_{sc})_{s\in S,c \in C}$ is a feasible solution of problem \eqref{opt_final}. Moreover,
we have $\sum_{s \in S} \sum_{c \in C} b_{sc}=\sum_{s \in S} \sum_{c \in C} {a}_{sc}{b}_{sc} = \sum_{s \in S} \sum_{c \in C} {z}_{sc} \leq O_2$,
%
Now taking the maximum of the LHS over all feasible solutions of problem \eqref{opt_semifinal}
we obtain $O_1 \leq O_2$.
%
Hence,  we have $O_1=O_2$.
The first part of the proof also shows how to construct
an optimal solution of~\eqref{opt_semifinal}
from that of~\eqref{opt_final}.
\end{proof}

\vspace{-0.3cm}
\begin{theorem}
\label{thm:np}
Problem \eqref{opt_final} is strongly NP hard.
\end{theorem}
The proof of Theorem~\ref{thm:np} is provided in Appendix~\ref{pf:np}.
It uses a reduction from the 3-partition problem, which is a known NP-hard problem.
\vspace{-0.2cm}

\section{Approximation algorithms}
\label{sec:algos}
\vspace{-0.1cm}

Since the joint allocation-matching
problem is NP-hard, an efficient algorithm 
for finding an exact solution is out of reach (unless $P=NP$).
We therefore look for allocation algorithms which 
provide approximate solutions and are easy to implement. 
Specifically, we 
consider the following allocation policies:

\subsubsection{The \texttt{greedy} policy}: 
The \texttt{greedy} algorithm computes a feasible replication policy $A \in \mc{A}$ as follows: 
A flow is assigned to each server $s\in S$ and each content $c \in C$
and are
denoted as $\mathsf{flow}(s)$ and $\mathsf{flow}(c)$, respectively. 
Additionally, each server $s \in S$
is also assigned a degree, denoted as $\mathsf{deg}(s)$.
Initially, we set $\mathsf{flow}(s)=U_s$, $\mathsf{deg}(s)=d_s$
for all $s \in S$ and $\mathsf{flow}(c)=\lambda_c$ for all $c \in C$.
Then, in each iteration, a pair $(s,c) \in S\times C$ is found for which
${\mathsf{flow}(s)\mathsf{deg}(s)\mathsf{flow}(c)} > 0$
 and which maximizes $\min\brac{\mathsf{flow}(s),\mathsf{flow}(c)}$. 
 If such a pair $(s^*,c^*)$
is found then the flow of both $s^*$ and $c^*$ are decreased
by an amount $\mathtt{MatchedFlow}=\min\brac{\mathsf{flow}(s^*),\mathsf{flow}(c^*)}$
and the degree of $s^*$ is reduced by one. Furthermore, 
we set $a_{s^*c^*}=1$ and $z_{s^*c^*}=\mathtt{MatchedFlow}$.
The iterations continue
until no such pair $(s^*,c^*)$ can be found. The pseudocode of the algorithm
is given as Algorithm \ref{alg:greedy}.

Clearly, the \texttt{greedy} algorithm terminates in at most 
$m+n$ iterations and returns a feasible replication policy in $\mc{A}$. 
Furthermore, in each iteration, the optimal pair
$(s^*,c^*)$ can be found in at most
$O(m+n)$ steps. Thus, the worst case time complexity of the
\texttt{greedy} algorithm is $O((m+n)^2)$.
We refer to the allocation policy computed by the greedy algorithm as the \texttt{greedy}
allocation policy.

\begin{algorithm}[H]
	\small{\textbf{Inputs}: {$\mc{U}=\brac{U_s, s\in S}$, $\Lambda=\brac{\lambda_c, c \in C}$, $\mc{D}=\brac{d_s,s \in S}$}}\\
	\textbf{Output}: {$Z=(z_{sc})$, $A=(a_{sc})$}\\
	\textbf{Initialize}: {$\mathsf{flow}(s) \gets U_s, \mathsf{deg}(s) \gets d_s, \forall s \in S$; $\mathsf{flow}(c) \gets \lambda_{c}, \forall c \in C$}
	\begin{algorithmic}[1]
		\While{$\exists (s,c) \in S \times C$ s.t ${\mathsf{flow}(s)\mathsf{deg}(s)\mathsf{flow}(c)} > 0$}
		\State $s^* \gets \arg\max_{s \in S: \mathsf{deg}(s) > 0} \brac{\mathsf{flow}(s)}$
		\State $c^* \gets \arg\max_{c \in C} \brac{\mathsf{flow}(c)}$
		\State $\mathtt{MatchedFlow} \gets\min\brac{\mathsf{flow}(s^*), \mathsf{flow}(c^*)}$
		\State $\mathsf{flow}(s^*) \gets \mathsf{flow}(s^*)-\mathtt{MatchedFlow}$,  $\mathsf{flow}(c^*) \gets \mathsf{flow}(c^*)-				\mathtt{MatchedFlow}$,  $\mathsf{deg}(s^*) \gets \mathsf{deg}(s^*)-1$
		\State $a_{s^*c^*} \gets 1$, $z_{s^*c^*} \gets \mathtt{MatchedFlow}$
		\EndWhile
	\end{algorithmic}
	\caption{\texttt{greedy}$(\mc{U}, \Lambda, \mc{D})$}\label{alg:greedy}
\end{algorithm}

We now show that the \texttt{greedy} algorithm always achieves
at least $1/2$ of the optimal value of problem~\eqref{opt_final}.
To state the result, we denote the instance of 
problem~\eqref{opt_final}, defined by the vector of bandwidth
capacities $\mc{U}=(U_s, s \in S)$, vector of arrival rates $\Lambda=(\lambda_c,c \in C)$, and 
the vector of memory sizes $\mc{D}=(d_s,s \in S)$ by $\mathtt{JAM}(\mc{U}, \Lambda, \mc{D})$
and its optimal value by $\mathtt{JAM}^*(\mc{U}, \Lambda,\mc{D})$. 
The following theorem, whose proof is given in Appendix~\ref{proof:approx},
provides a performance guarantee for the \texttt{greedy} algorithm.


\begin{theorem}
\label{thm:approx}
The output of \texttt{greedy} 
on $\mathtt{JAM}(\mc{U},\Lambda,\mc{D})$ is at least 
$\frac{1}{2}\mathtt{JAM}^*(\mc{U},\Lambda,\mc{D})$
\end{theorem}



\subsubsection{Proportional to product (\texttt{p2p}) policy}

Under the \texttt{p2p} policy, a server with memory size $d_j$, $j \in \mc{J}$,
is allocated all the contents belonging to a set $K \subseteq C$ of size $d_j$ with probability 

\begin{equation}
p_{jK}= \frac{1}{Z_j} \prod_{c \in K} \hat{\lambda}_c,
\label{eq:p2p}
\end{equation}
independently of all other servers in the system, where 
$\hat{\lambda}_c=\lambda_c/\sum_{c' \in C} \lambda_{c'}$
denotes the normalized arrival rate of content $c$ and 
$Z_j=\sum_{K \subseteq C, \abs{K}=d_j} \prod_{c \in K} \hat{\lambda}_c,$. 
This scheme was proposed in~\cite{Tan_massoulie_Caching}.
Unlike the \texttt{greedy} policy, it does not take into account the bandwidths
of the servers. Furthermore, it is difficult to provide any 
performance guarantee for this scheme.
Nevertheless, we analyze the dynamics
and evaluate the performance of the system under the \texttt{p2p} policy later in the paper. 

\subsubsection{Uniform (\texttt{unif}) policy}

As a baseline for comparison, we consider a naive strategy for
replication where a server with memory size $d_j$, $j\in \mc{J}$,
is populated by all the contents of the set $K \subseteq C$ of size $d_j$ with
probability

\begin{equation}
p_{jK}= \frac{1}{\binom{m}{d_j}}
\label{eq:unif}
\end{equation}
Note that~\eqref{eq:unif} does not depend on the particular set $K$
and is the same for all $K$ of the same size.
Furthermore, \eqref{eq:unif} follows from \eqref{eq:p2p} if 
$\hat{\lambda}_c=1/m$ for all $c \in C$. Thus, the \texttt{unif}
policy treats all contents to be equally popular.
It is easy to implement in practice since it does not require the knowledge
of the popularities of the contents.

\section{Dynamics of the system: Asymptotic optimality}
\label{sec:grand}

We now analyze the dynamics of the system when 
the caches have been populated by one of the algorithms discussed in the previous section.  
To describe the dynamics of the caching system, we need also to specify 
the matching scheme used to assign
each incoming request to one of the servers.
In previous works e.g. \cite{Tan_massoulie_Caching,Lelarge_caching}
the maximum matching algorithm has been considered. 
In this algorithm, an incoming request 
is accepted if a matching $B=(b_{sc})_{s \in S, c \in C}$
satisfying the constraints~\eqref{eq:cons_u} and \eqref{eq:cons_a} can be found
such that $\sum_{s \in S, c \in C} b_{sc}$ equals the total number of requests in the system
including the new request.
Such a matching, although optimal, is hard to compute in practice 
and  may potentially involve {\em repacking} or reassignment
of the ongoing requests to other available servers. 
Such repacking would cause undesirable interruptions of service and may lead to increased delay.  We therefore look for 
matching policies which do not involve repacking of the ongoing requests. 
We present next a simpler matching strategy.

\subsection{Random Available Server (RAS) matching policy}

In this policy, each newly arrived content request
is assigned to a server chosen uniformly at random from
the set of all servers that store the content and are able
to serve an additional request of the content. 
If no such server is available, then the request is blocked. 
This scheme can be implemented by maintaining
a list of available servers for each content at a central job 
dispatcher that makes the matching decisions. Note that
it is not necessary for the central job dispatcher to keep track of the number of jobs
in each server. It is sufficient to just keep track of whether a server is operating at
its maximum bandwidth capacity or not.
This can be achieved by sending a message from a server back to the job dispatcher
whenever a job leaves the server previously operating at its maximum bandwidth capacity.
Since such updates occur in the background, when a new request arrives, 
it can be immediately matched to an available server. 
Unlike the maximum matching algorithm, 
the RAS scheme does not cause repacking of existing requests  in the system.  
Next, we analyze the dynamics of the system under the RAS matching policy
under specific assumptions on the arrival and service time processes for large system sizes.

\subsection{Large system asymptotics}

We assume that 
the requests of each content $c$ arrive according to a Poisson process with rate $\lambda_c$,
independent of all other processes. 
Furthermore, the service time of each request is assumed to be exponentially distributed
with unit mean.\footnote{We later show numerically that our results 
do not depend on the type of service time distribution.}
Let $\rho:=\sum_{c \in C} \lambda_c/n\sum_{i \in \mc{I}, j \in \mc{J}} \alpha_{ij} U_i$
denote the load on the system.

We are specifically interested in an asymptotic scaling regime where the number of servers $n$ 
goes to infinity keeping the load $\rho$ and the proportions $\alpha_{ij}$, $i \in \mc{I}, j \in \mc{J}$ fixed. 
This is achieved by keeping the same catalog $C$ of contents but scaling 
the arrival rate of each content linearly with $n$, i.e., $\lambda_{c}=n\bar\lambda_{c}$ for all $c \in C$.
We define $\bar{\lambda}=\sum_{c \in C} \bar\lambda_c$.
Note that the normalized arrival rates $\hat{\lambda}_c$ remains the same.
This represents a scenario where the system size scales with the demands of 
the contents.


Before analyzing the dynamics of the caching system, 
we first determine the fraction of servers in a particular cache configuration
under a given allocation policy in the limiting system.
Under a given allocation policy,  let $q_{ijK}^{(n)}$ denote the fraction
of servers having
bandwidth $U_i$ ($i \in \mc{I}$) and  memory size $d_j$ ($j \in \mc{J}$) that are storing 
all the contents belonging to the set $K \subseteq C$ in the $n$th system.
We call these servers to be in {\em configuration} $(i,j,K)$ and denote the set of all
such configurations as $\mc{L}$.
Define $q_{ijK}:=\lim_{n\to \infty} q_{ijK}^{(n)}$ for each $(i,j,K) \in \mc{L}$, if the limit exists.
Clearly, for the \texttt{p2p} and the \texttt{unif} allocation policies $q_{ijK}^{(n)}=q_{ijK}=p_{jK}$,
where $p_{jK}$ is defined in \eqref{eq:p2p} and \eqref{eq:unif}, respectively.
The following lemma shows that $q_{ijK}$ exists for the \texttt{greedy} policy and further characterizes it.  
We first denote by $q_{ijc}$ the value of $q_{ijK}$ for $K=\cbrac{c}$ and 
define $\theta_c:=\sum_{i \in \mc{I}, j\in \mc{J}} \alpha_{ij} U_i q_{ijc}$ to be the normalized capacity allocated to content $c$.  
The lemma then states that for $\rho<1$ 
the capacity allocated to a content is equal to the arrival rate of the content, 
and for $\rho>1$, unpopular contents are allocated zero capacity.  
The proof of the lemma is given in Appendix~\ref{proof:greedy}.

\begin{lemma}
\label{lem:greedy}
Under the \texttt{greedy} allocation policy, we have 
$q_{ijK}=0$ for $\abs{K} \geq 2$.
Furthermore, for $\rho \leq 1$ we have 
$\theta_c=\bar\lambda_c$, $\forall c \in C$.
%
%
For $\rho > 1$, we have the following: 
If the popularities are ordered as $\bar\lambda_1> \bar\lambda_2>\ldots>\bar\lambda_m > 0$
and with $c^*$  such that $\frac{\bar{\lambda}}{\rho}  \in \left(\sum_{c'=1}^{c^*-1} \right.\bar\lambda_{c'} -(c^*-1)\bar\lambda_{c^*}, \left. \sum_{c'=1}^{c^*} \bar\lambda_{c'} -c^*\bar\lambda_{c^*+1}\right]$ then 
\begin{equation}
\theta_c= \begin{cases}
                   \bar{\lambda}_c- \frac{1}{c^*}\sbrac{\sum_{c'=1}^{c^*} \bar\lambda_{c'}-\frac{\bar{\lambda}}{\rho}} & \text{ if } 1\leq c\leq c^*\\
                   0 & \text{ if } c^*+1 \leq c \leq m
                   \end{cases}
\label{eq:theta2}
\end{equation}
\end{lemma}

Let $x^{(n)}_{i,j,K,r}(t)$, for $r \in [0, U_i]$,
denote the fraction of servers in configuration $(i,j,K) \in \mc{L}$ serving at least
$r$ requests at time $t \geq 0$.
Clearly, the process $x^{(n)}=(x^{(n)}_{i,j,K,r}(t), (i,j,K) \in \mc{L}, 0\leq r \leq U_i, t \geq 0)$
is Markovian and takes values in the set $\mc{W}$, where
\begin{multline*}
\mc{W}:=\left\{w=\brac{w_{i,j,K,r}, (i,j,K) \in \mc{L}, r \in \mb{Z}_+}:\right.1=w_{i,j,K,0}\\
\geq w_{i,j,K,1}\geq \ldots \geq w_{i,j,K,U_i} \geq 0=w_{i,j,K,r},
\left. \forall r > U_i \right\}.
\end{multline*}
%
It is easy to see that according to the RAS policy, the Markov process $x^{(n)}$ jumps
from a given state $w \in \mc{W}$ to the state $w+e_{i,j,K,r}/n\alpha_{ij}q_{ijK}^{(n)} \in \mc{W}$ ($r \geq 1$)
with rate $$\sum_{c \in K} n \bar{\lambda}_c \frac{\alpha_{ij} q^{(n)}_{ijK} (w_{i,j,K,r-1}-w_{i,j,K,r})\indic\brac{w_{i,j,K,U_i} < 1}}{\sum_{K': c \in K'} \sum_{i,j} \alpha_{ij}q^{(n)}_{ijK'}(1-w_{i,j,K',U_i})}, 
$$
%
where $e_{i,j,K,r}$ denotes the unit vector with unity in position $(i,j,K,r)$.
The transition corresponds to an arrival of a request for content $c \in K$.
Similarly, the rate of downward transition from $w \in \mc{W}$ to $w-e_{i,j,K,r}/n\alpha_{ij}q_{ijK}^{(n)}$
can be computed to be $rn \alpha_{ij} q^{(n)}_{ijK}(w_{i,j,K,r}-w_{i,j,K,r+1})$.

We are interested in the limiting behavior of the process $x^{(n)}$ as $n \to \infty$.
We first notice from its transition rates, that $x^{(n)}$
is a {\em density dependent jump Markov process}~\cite{Kurtz_ODE, Mitzenmacher_thesis,ArpanSSY} with a limiting ($n \to \infty$)
{\em conditional drift} given by the mapping $h:=(h_{i,j,K,r}, (i,j,K) \in \mc{L}, r \in \mb{Z}_+)$ on $\mc{W}$,
defined as  $h_{i,j,K,r}(w)=0$ for $r \in \cbrac{0} \cup [U_i+1, \infty)$ and
\begin{align}
&h_{i,j,K,r}(w)=\sum_{c \in K} \bar{\lambda}_c \frac{(w_{i,j,K,r-1}-w_{i,j,K,r})}{\sum_{K': c \in K'} \sum_{i,j} \alpha_{ij}q_{ijK'}(1-w_{i,j,K',U_i})}
\nonumber\\ & \times \indic\brac{w_{i,j,K,U_i} < 1} 
-r(w_{i,j,K,r}-w_{i,j,K,r+1}),  \text{ for } r \geq 1, \label{eq:ode1}
\end{align}
for each $(i,j,K) \in \mc{L}$. 
For systems in which the limiting drift $h$ is Lipschitz continuous, the classical results
of Kurtz~\cite{Kurtz_ODE} imply that the process $x^{(n)}$ converges in distribution 
to the unique deterministic process $x=(x(t), t \geq 0)$ satisfying $\dot{x}=h(x)$. The process $x$ is called the {\em mean field limit}
or the {\em fluid limit} of the system.
However, since in our case the RHS of~\eqref{eq:ode1}
has discontinuities (due to the presence of the indicator terms), the solution to $\dot{x}=h(x)$ is not well defined. 
To overcome this, we define a process
$x$ as the solution of a {\em differential inclusion} (DI)
given as $\dot{x}\in H(x)$, where $H$ is set valued mapping on $\mc{W}$
defined as the cartesian product of set valued maps $H_{i,j,K,r}$ over all $(i,j,K,r)$.
For each $(i,j,K) \in\mc{L}$ we define
$H_{i,j,K,r}(w)=\cbrac{0}$ for $r \in \cbrac{0} \cup [U_i+1, \infty)$ and 
\begin{align}
&H_{i,j,K,r}(w)=\sbrac{0, \sum_{c \in K} \frac{\bar\lambda_{c}}{\alpha_{ij}q_{ijK}}} \indic\brac{w_{i,j,K,U_i} = 1}  \nonumber\\
&+\sum_{c \in K} \bar{\lambda}_c \frac{(w_{i,j,K,r-1}-w_{i,j,K,r})}{\sum_{K': c \in K'} \sum_{i,j} \alpha_{ij}q_{ijK'}(1-w_{i,j,K',U_i})}
\nonumber\\ & \times \indic\brac{w_{i,j,K,U_i} < 1}  -r(w_{i,j,K,r}-w_{i,j,K,r+1}), \text{ for } r \geq 1 \label{eq:di1}.
\end{align}
We note that that the set $H_{i,j,K,r}(w)$ is the convex hull of all the limit points of $h_{i,j,K,r}(w)$.
In the next theorem, whose proof is given in Appendix~\ref{proof:mfconv},
we show that the DI $\dot{x} \in H(x)$ has well defined solutions and the process
$x^{(n)}$ converges to one of the it solutions $n \to \infty$ in probability.
\begin{theorem}
	\label{thm:mfconv}
	For any $x_0 \in \mc{W}$, the set $\mc{S}_{x_0}$ of solutions to
	the DI $\dot{x} \in H(x)$ with $x(0)=x_0$ is non-empty.
	Furthermore, if $x^{(n)}(0) \xrightarrow{p} x_0 \in \mc{X}$ as $n \to \infty$, then for all $T > 0$ we have 
	$\inf_{x \in \mc{S}_{x_0}}\sup_{t \in [0,T]} \norm{x^{(n)}(t)-x(t)} \xrightarrow{p} 0$ as $n \to \infty$, where the process $x$ is a solution of
	the DI $\dot{x} \in H(x)$.
\end{theorem}  

Thus far we have seen that the limiting dynamics of the system
for any finite time $t \geq 0$ is described by a solution of the DI $\dot{x} \in H(x)$. 
We are also interested in the stationary behavior of the limiting system, i.e.,
the behavior of $x^{(n)}(t)$ as both $n \to \infty$ and $t \to \infty$.\footnote{For finite $n$, the system always 
reaches stationarity because the process $x^{(n)}$
is irreducible on a finite state space.}
Specifically, we are interested in the total number of jobs 
processed by the system in the stationary regime.
Let  $Y^{(n)}(t)=n \sum_{i,j,K} \alpha_{ij} q_{ijK}^{(n)}\sum_{r=1}^{U_i} x_{i,j,K,r}^{(n)}(t)$ be the total number of
requests in the $n$th system at time $t \geq 0$ and 
let  $Y^{(n)}(\infty)$ denote its random
stationary value. Define $y^{(n)}(t):=Y^{(n)}(t)/n$ for all $t \in [0,\infty]$ and $y(t):=\sum_{i,j,K} \alpha_{ij} q_{ijK}\sum_{r=1}^{U_i} x_{i,j,K,r}(t)$.
In the next theorem, whose proof is provided in Appendix~\ref{proof:global_stability}, we show that
under the \texttt{greedy} and the \texttt{p2p} combined with the RAS matching policy, $y(\infty)=\lim_{t \to \infty} y(t)=\min\brac{\bar{\lambda}, \bar{\lambda}/\rho}$ and $\lim_{n \to \infty} y^{n}(\infty)=y(\infty)$.


\begin{theorem}
\label{thm:global_stability}
Under the \texttt{greedy} allocation policy combined with the RAS matching policy or under the \texttt{p2p} 
allocation policy with $d_s=1$ for all $s \in S$ combined with the RAS matching policy, 
we have $y(\infty)=\lim_{t \to \infty} y(t)=\min(\bar{\lambda}, \bar\lambda/\rho)$.
Furthermore, the sequence $(y^{(n)}(\infty))_n$ is tight and  
$y^{(n)}(\infty) \xrightarrow{p} y(\infty)$.
\end{theorem}

\subsection{Asymptotic optimality}

The above theorem establishes that the \texttt{p2p} and \texttt{greedy}
the allocation schemes when combined with the RAS matching policy 
are optimal in the limiting system.
To see this, we find an upper bound on ${Y}^{(n)}(t)$ for any 
combination of allocation scheme and matching scheme.
A trivial upper bound on ${Y}^{(n)}(t)$
is clearly the total bandwidth capacity of the system, i.e., ${Y}^{(n)}(t) \leq n\sum_{ij}\alpha_{ij} U_i=n\bar{\lambda}/\rho$ for all $t \in [0,\infty]$.
Another upper bound on $Y^{(n)}(\infty)$ can be obtained by comparing
the system with an hypothetical caching system in which each server has infinite bandwidth and 
each content is stored in at least in one server. Clearly, this system behaves as an $M/M/\infty$ system
serving all incoming requests. Hence, the stationary number  of requests $\bar{Y}(\infty)$
in this hypothetical system is a Poisson random variable with mean $n \bar\lambda$.
By a simple coupling argument, it follows that $Y^{(n)}(\infty) \leq \bar{Y}(\infty)$ almost surely for all $n$. 
Thus, combining both upper bounds we have $y^{(n)}(\infty) \leq \min\brac{\bar{Y}(\infty)/n, \bar\lambda/\rho}$. 
Hence, $\limsup_{n \to \infty} y^{(n)}(\infty) \leq \min\brac{\bar{\lambda},\bar\lambda/\rho}$. 
But Theorem~\ref{thm:global_stability} shows that for the proposed schemes $\lim_{n \to \infty} y^{(n)}(\infty) = \min\brac{\bar{\lambda},\bar\lambda/\rho}$.
Hence, the proposed schemes are asymptotically optimal. It is easy to see that the corresponding optimal (minimal)
blocking probability is given by $(1-\min(1,1/\rho))=(1-1/\rho)^+$.

\section{Numerical Results}
\label{sec:numerics}

We now numerically evaluate the performance of the caching system
under the various allocation policies combined with the RAS matching policy for finite system sizes.
Specifically, we consider a system with $d_s=2, U_s=1 \forall s \in S$,
$m=500$. 
The popularities of the contents are chosen
according to a Zipf like distribution~\cite{Zipf_study, Moharir_caching}, where 
the normalized arrival rate $\hat{\lambda}_c$ 
of any content $c \in C=\cbrac{1,2,\ldots,m}$ is chosen to be
$\hat{\lambda}_c= c^{-\eta}/\sum_{c' \in C}( c')^{-\eta}$ for $\eta =2$. 
The system is simulated for different values
of $n$ and $\rho$ for 160000 arrivals. In Figures~\ref{fig:aRAS} and~\ref{fig:bRAS}, we plot the stationary blocking probability 
of a request as a function the number of servers for different allocation policies
combined with the RAS matching policy for $\rho=0.8$ and $\rho=1.2$, respectively. 
We observe that under both \texttt{greedy}
and \texttt{p2p} policies the blocking probability approaches the optimal lower bound $(1-1/\rho)^+$
as $n$ increases. In Table~\ref{tab:conv}, we show the difference between the blocking probability
of the finite system $P_{\text{blocking}}^{(n)}$ and the optimal lower bound $P_{\text{opt}}=(1-1/\rho)^+$
as a function of the system size $n$ for the \texttt{p2p} algorithm for $\rho=0.8$.
We  observe that the distance decreases as $O(n^{-1/2})$.
The same is observed for the \texttt{greedy} policy.
Such rate of convergence is in accordance with the recent 
results in the literature on mean field convergence~\cite{Ying_conv}. \vspace{-0.2cm}
\begin{figure}
	\centering     
	\subfigure[$\rho=0.8$]{\label{fig:aRAS}\includegraphics[scale=0.39]{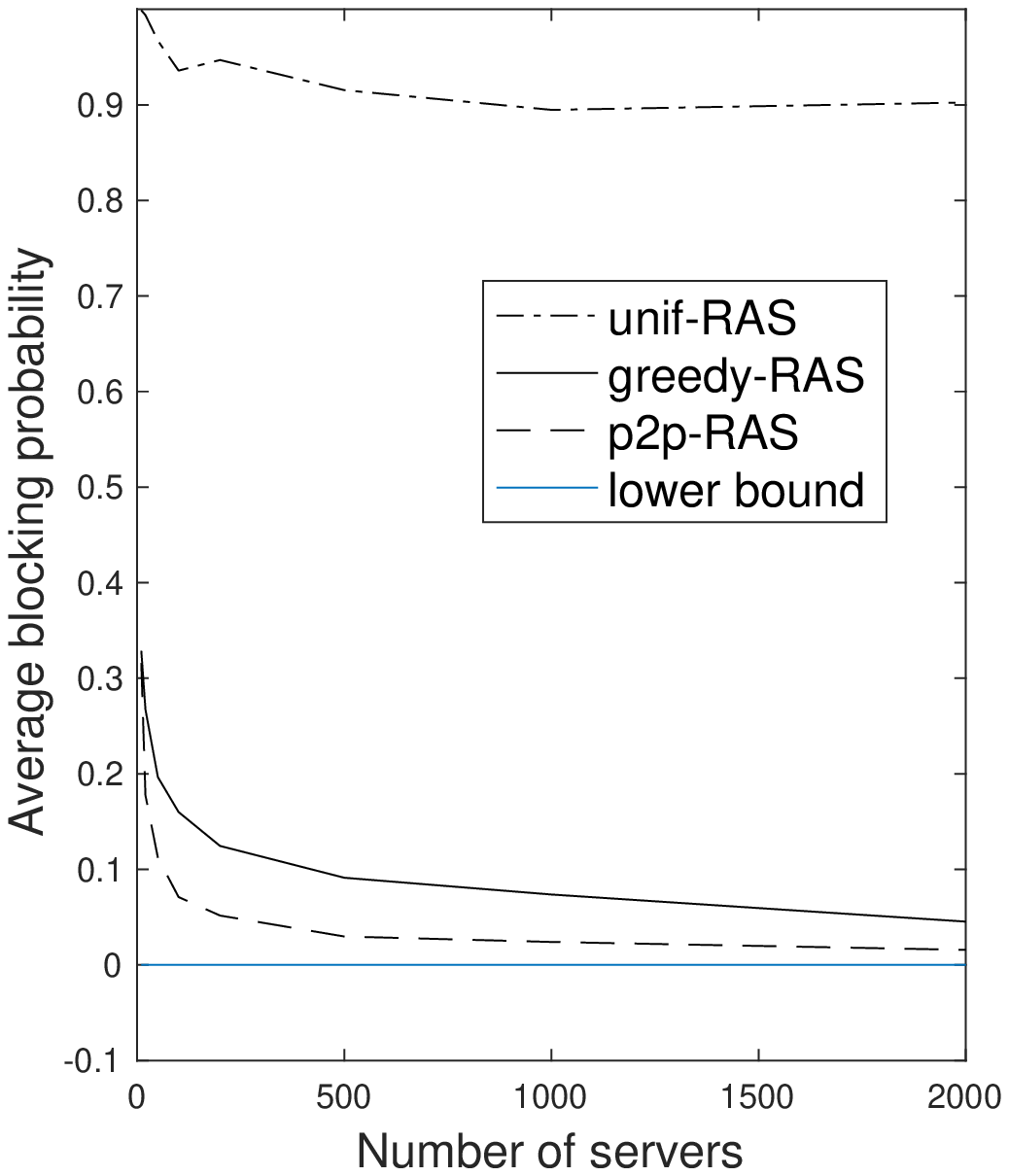}}
	\subfigure[$\rho=1.2$]{\label{fig:bRAS}\includegraphics[scale=0.39]{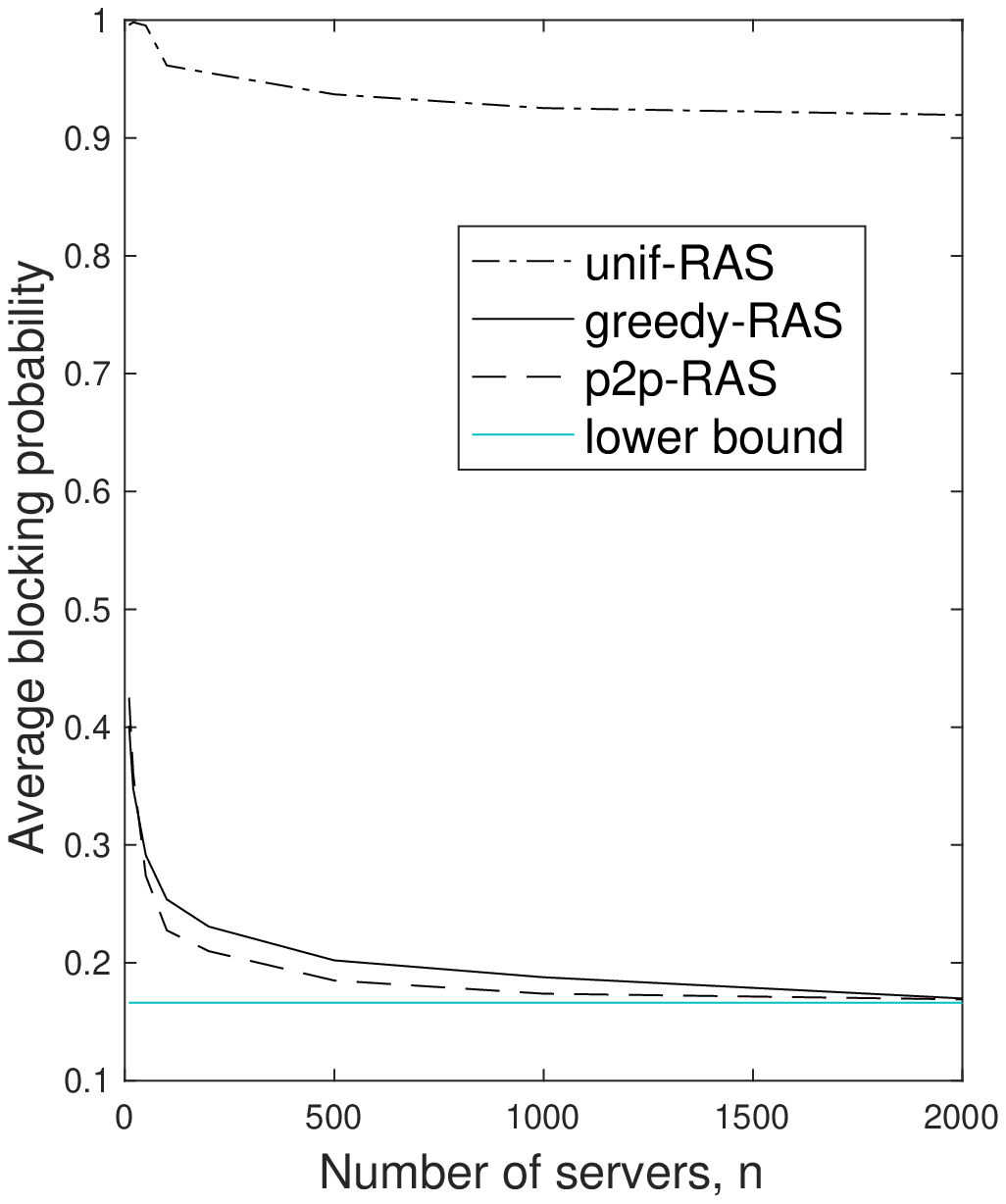}}
	\caption{Average blocking probability as a function of $n$}\label{fig:n_loss_ras}
\end{figure}
\begin{table}
	\renewcommand{\arraystretch}{0.9}
	\caption{Convergence of $P_{\text{blocking}}^{(n)}$}
	\centering
		\begin{tabular}{|c  | c |}
			\hline
			$n$&  $\abs{P_{\text{blocking}}^{(n)}-P_{\text{opt}}}$\\
			\hline
			10 & 0.2612\\
			20 & 0.1837\\
			50 & 0.1130\\
			200& 0.0501\\
			1000& 0.0220\\
			2000& 0.0148\\
			\hline
		\end{tabular}%
	\label{tab:conv}
\end{table}

The blocking probability as a function of the load $\rho$ is shown in Figure~\ref{fig:rho_loss} for $n=400$.
We observe that for the given parameter setting the \texttt{p2p} policy
performs better than the \texttt{greedy} policy, but the difference is not significant at high loads. In Figure~\ref{fig:eta_loss},
we plot the average blocking probability of requests as a function of the decay factor
$\eta$ of the popularity distribution for $\rho=0.8$ and $n=400$.
We observe that as $\eta$ increases and the popularity distribution becomes more skewed towards 
higher popularity contents, the performance of the \texttt{unif} policy degrades
as it still treats all contents to be equally popular. On the other hand,
for both \texttt{greedy} and \texttt{p2p} policies the performance improves. This implies
that in these policies higher popularity contents are given more priority than
lower popularity contents.

%

\begin{figure*}[t!]
    \centering
    \subfigure[Load]{\label{fig:rho_loss}\includegraphics[width=0.33\textwidth]{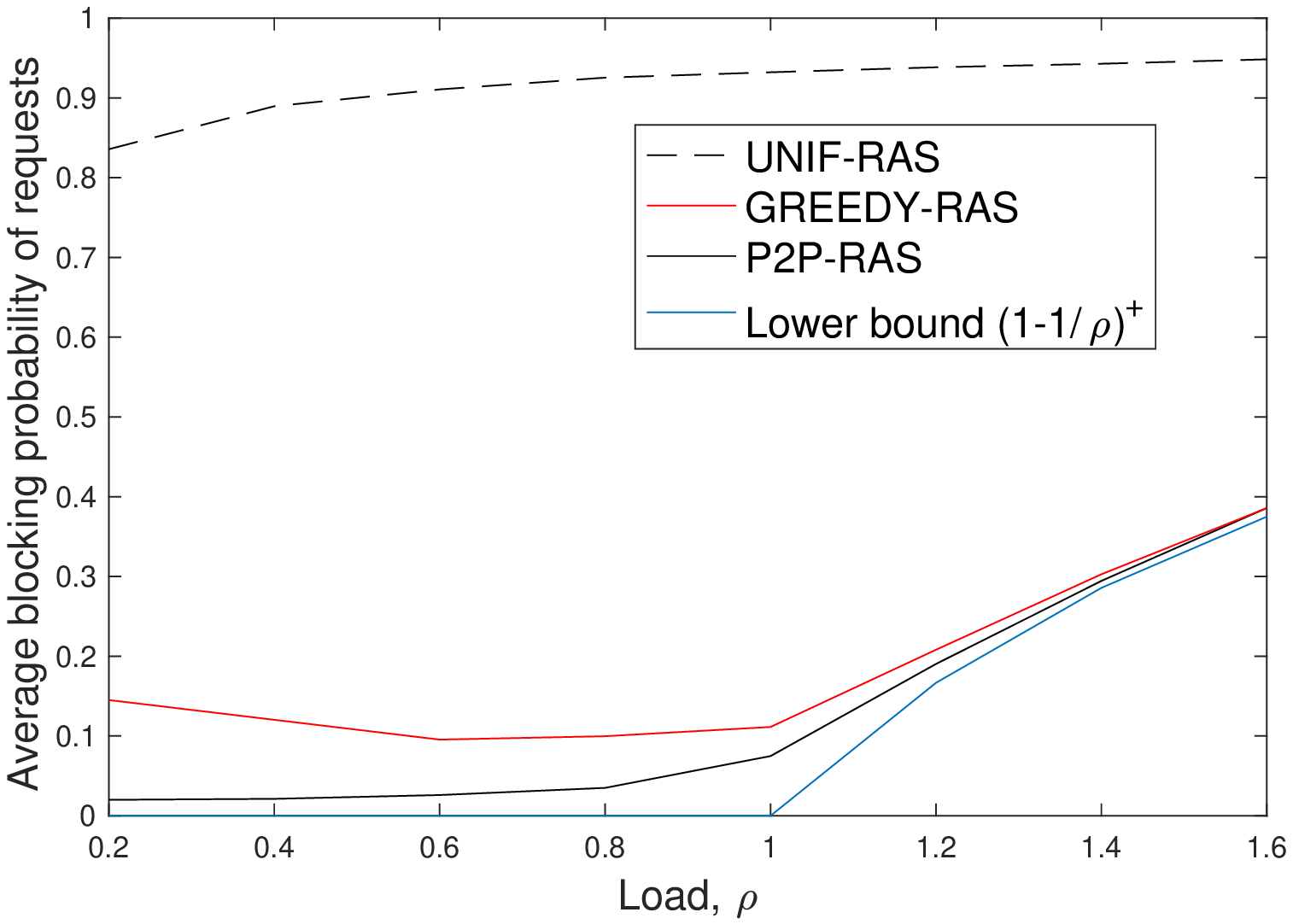}}
   \subfigure[Popularity decay factor]{\label{fig:eta_loss}\includegraphics[width=0.33\textwidth]{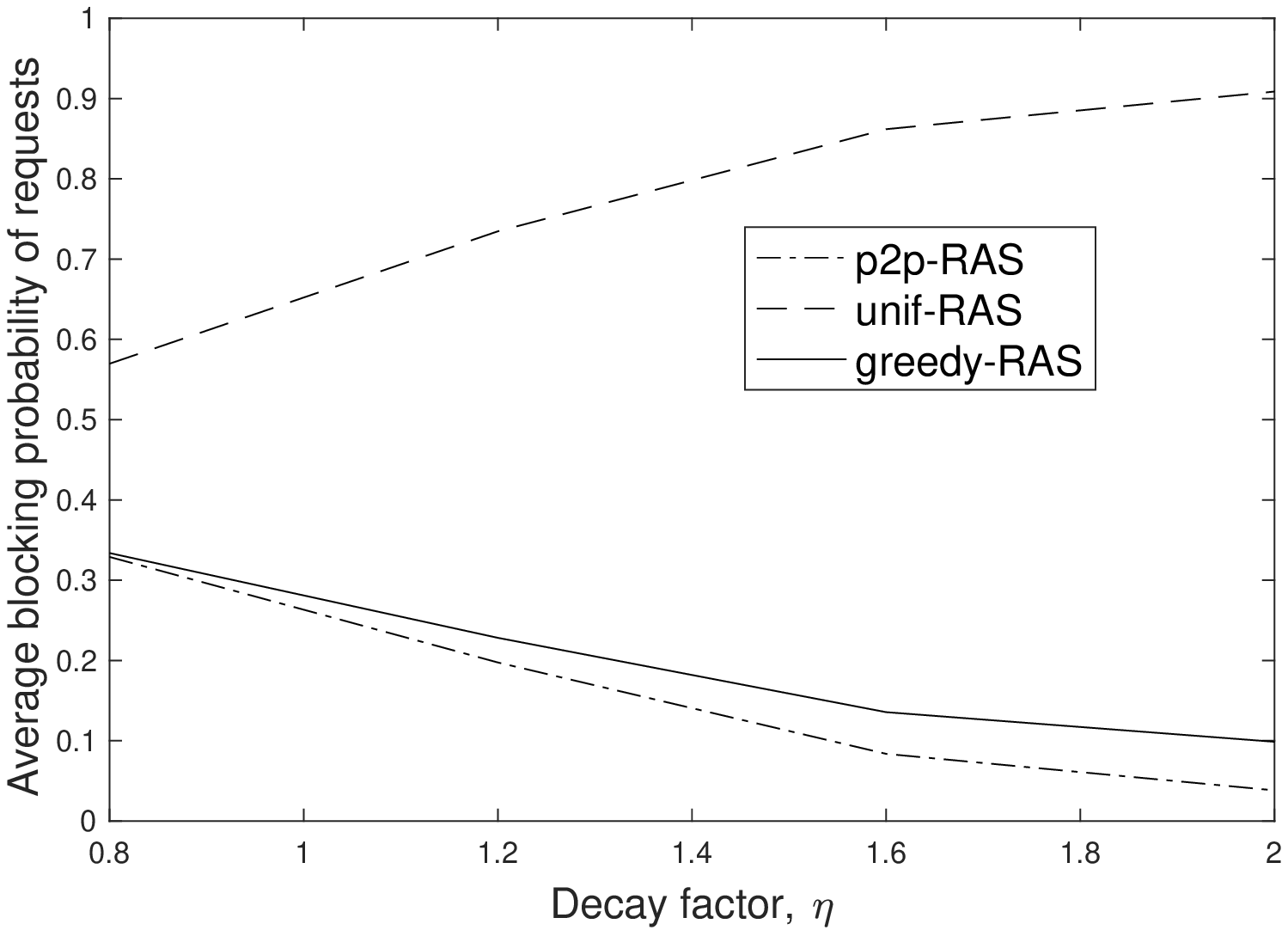}}
   \subfigure[Number of contents]{\label{fig:content_loss}\includegraphics[width=0.32\textwidth]{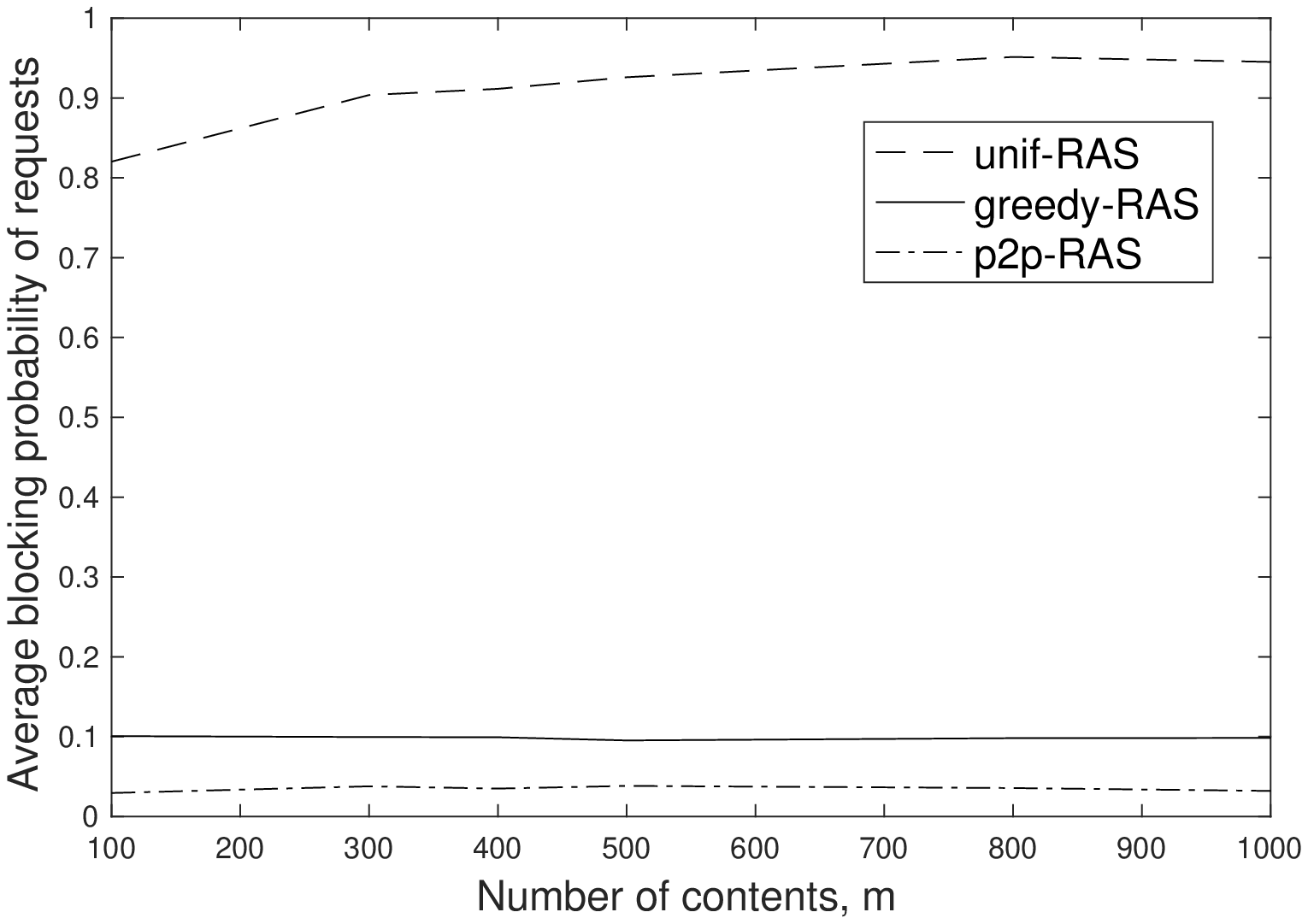}}
   \caption{Average blocking probability as a function of (a) the load $\rho$, (b) the popularity decay factor $\eta$, and (c)
   	the number of contents $m$, for the various allocation schemes.}
\end{figure*}

Next, we study the sensitivity of the system to the type of distribution of the service times
of the requests. For this purpose, we consider the following types of service time distribution with unit mean:
Exponential, Constant, Lognormal with probability density function (PDF) 
given by $f(x)=(1/x\sqrt{2\pi}) e^{-(\ln x+0.5)^2/2}$, and Pareto with PDF given by
$f(x)=10(0.9)^{10}/x^{11}$. For each type of distribution 
we simulate the system (for 160000 arrivals) with the same parameters as described above
under the \texttt{greedy} algorithm (combined with the RAS matching). The blocking probability of requests is tabulated 
as a function of $\rho$ in Table~\ref{tab:insense} for different distributions.
We observe that for the same value of $\rho$ the blocking probabilities are nearly the same for
all distributions. The values become even more closer when the system size is increased.
This seems to suggest the system approaches {\em insensitivity} as $n \to \infty$.
Such asymptotic insensitivity is known to hold for similar systems~\cite{Arpan_perform,thiru_loss_arxiv}. However, a proof
remains  an open problem.

Finally, in Figure~\ref{fig:content_loss} we plot the average blocking probability 
of the requests as a function of the number of contents for $\rho=0.8$, $\eta=2$,
and $N=400$. We observe that for the \texttt{p2p} and the \texttt{greedy}
policies the average blocking probability remains almost constant with the variation 
of the number of contents. This is because even though the total number of contents is increasing,
only a small fraction of them are highly popular. As a result the addition of more contents does not affect the
performance of the system. This also justifies why we scale the arrival rates of the contents
instead of the number of contents in Section~\ref{sec:grand}.
%
\begin{table}
	\renewcommand{\arraystretch}{0.9}
	\caption{Sensitivity to the type of service time distributions}
	\centering
		\begin{tabular}{| c | c | c | c | c |}
			\hline
			$(n,\rho)$&  Exponential & Constant & Lognormal & Pareto\\
			\hline
			(400, 0.4) & 0.1200 & 0.1219 & 0.1241  &0.1207\\
			(400, 0.8) & 0.0989  & 0.1028 & 0.1006 & 0.0999\\
			(400, 1.2) & 0.2084 & 0.2084  & 0.2110 & 0.2092\\
			(400, 1.6) & 0.3863  & 0.3854  & 0.3878 & 0.3874\\
			(1000, 0.4) & 0.0916  & 0.0912  & 0.0915 & 0.0919\\
			\hline
		\end{tabular}%
	\label{tab:insense}
\end{table}


\section{Conclusions}
\label{sec:conclusion}

We considered the joint problem of 
content placement and request matching 
in a distributed network of content servers.  
We formulated the problem in an optimization framework 
and showed that it is NP-hard.  
We then presented a polynomial-time greedy algorithm that approximates to 
a constant of the optimal value.  
We then considered a large systems scaling regime, where we showed that a simple greedy matching policy is asymptotically optimal.
We employed a new approach based on the theory of differential inclusions to prove the fluid limit results.
Many interesting avenues of future work exist.
One such challenge is to find the combination of optimal
allocation and matching algorithms for systems
where both the number of servers and the number of contents
scale proportionally to each other.


%
\bibliographystyle{unsrt}
\bibliography{caching}  

\begin{thebibliography}{10}

\bibitem{CiscoWP17}
Cisco.
\newblock Cisco visual networking index: Forecast and methodology,
  2016--€"2021.
\newblock Technical report, Cisco, 2017.

\bibitem{Dilley2002}
J.~Dilley, B.~Maggs, J.~Parikh, H.~Prokop, R.~Sitaraman, and B.~Weihl.
\newblock Globally distributed content delivery.
\newblock {\em IEEE Internet Computing}, 6(5):50--58, 2002.

\bibitem{Kelly_loss}
F.~Kelly.
\newblock Loss networks.
\newblock {\em Ann. Appl. Prob.}, 1(3):319--378, 1991.

\bibitem{Tan_massoulie_Caching}
B.~Tan and L.~Massouli\'e.
\newblock Optimal content placement for peer-to-peer video-on-demand systems.
\newblock {\em IEEE/ACM Transactions on Networking}, 21(2):566--579, April
  2013.

\bibitem{Lelarge_caching}
M.~Leconte, M.~Lelarge, and L.~Massouli\'e.
\newblock Bipartite graph structures for efficient balancing of heterogeneous
  loads.
\newblock {\em SIGMETRICS Perform. Eval. Rev.}, 40(1):41--52, June 2012.

\bibitem{Kleinrock_caching}
S.~Tewari and L.~Kleinrock.
\newblock Proportional replication in peer-to-peer networks.
\newblock In {\em Proc. IEEE INFOCOM}, 2006.

\bibitem{wireless_caching}
N.~Golrezaei, K.~Shanmugam, A.~G. Dimakis, A.~F. Molisch, and G.~Caire.
\newblock Femtocaching: Wireless video content delivery through distributed
  caching helpers.
\newblock {\em IEEE Trans. Inf. Theory}, 59(12), 2013.

\bibitem{Anand_caching}
M.~Dehghan, A.~Seetharam, B.~Jiang, T.~He, T.~Salonidis, J.~Kurose, D.~Towsley,
  and R.~Sitaraman.
\newblock On the complexity of optimal routing and content caching in
  heterogeneous networks.
\newblock In {\em Proc. IEEE INFOCOM}, 2015.

\bibitem{Moharir_caching}
S.~Moharir and N.~Karamchandani.
\newblock Content replication in large distributed caches.
\newblock arXiv: 1603.09153 [cs.NI].

\bibitem{Stolyar_loss}
A.~L. Stolyar.
\newblock Large-scale heterogeneous service systems with general packing
  constraints.
\newblock {\em Adv. Appl. Prob.}, 49(1), 2017.

\bibitem{femto_caching}
N.~Golrezaei, K.~Shanmugam, A.~G. Dimakis, A.~F. Molisch, and G.~Caire.
\newblock Femtocaching: Wireless video content delivery through distributed
  caching helpers.
\newblock In {\em Proc. IEEE INFOCOM}, pages 1107--1115, 2012.

\bibitem{Moharir_sigm_2014}
S.~Moharir, J.~Ghaderi, S.~Sanghavi, and S.~Shakkottai.
\newblock Serving content with unknown demand: The high-dimensional regime.
\newblock {\em SIGMETRICS Perform. Eval. Rev.}, 42(1):435--447, June 2014.

\bibitem{Kurtz_ODE}
T.~G. Kurtz.
\newblock Solutions of ordinary differential equations as limits of pure jump
  markov processes.
\newblock {\em Journal of Applied Probability}, 7(1):49--58, 1970.

\bibitem{Mitzenmacher_thesis}
M.~Mitzenmacher.
\newblock {\em The power of two choices in randomized load balancing}.
\newblock PhD thesis, University of California at Berkeley, 1996.

\bibitem{ArpanSSY}
A.~Mukhopadhyay, A.~Karthik, and R.~R. Mazumdar.
\newblock Randomized assignment of jobs to servers in heterogeneous clusters of
  shared servers for low delay.
\newblock {\em Stochastic Systems}, 6(1):90--131, 2016.

\bibitem{Zipf_study}
L.~Breslau, P.~Cao, L.~Fan, G.~Phillips, and S.~Shenker.
\newblock Web caching and zipf-like distributions: evidence and implications.
\newblock In {\em Proc. IEEE INFOCOM}, 1999.

\bibitem{Ying_conv}
L~Ying.
\newblock On the approximation error of mean-field models.
\newblock {\em SIGMETRICS Perform. Eval. Rev.}, 44(1):285--297, June 2016.

\bibitem{Arpan_perform}
A.~Mukhopadhyayay, A.~Karthik, R.~R. Mazumdar, and F.~M. Guillemin.
\newblock Mean field and propagation of chaos in multi-class heterogeneous loss
  models.
\newblock {\em Perform. Eval.}, 91:117--131, September 2015.

\bibitem{thiru_loss_arxiv}
T.~{Vasantam}, A.~{Mukhopadhyay}, and R.~R {Mazumdar}.
\newblock {Insensitivity of the mean-field Limit of Loss Systems Under
  Power-of-d Routing}.
\newblock arXiv: 1708.09328, August 2017.

\bibitem{Gast_SIG}
N.~Gast and B.~Gaujal.
\newblock Mean field limit of non-smooth systems and differential inclusions.
\newblock {\em SIGMETRICS Perform. Eval. Rev.}, 38(2):30--32, October 2010.

\end{thebibliography}
%
%
\appendices

\section{Proof of Theorem~\ref{thm:np}}
\label{pf:np}

We prove this by reducing the 3-partition problem
to a decision problem version of~\eqref{opt_final}.
This will show that the optimization version of \eqref{opt_final} is NP-hard.
The 3-partition problem is defined as follows:
Given a finite set $G$ of
$3n$ elements, a number $L > 0$ and a mapping $\mathsf{size}: G \to \mb{R}_{++}$
 satisfying $\sum_{g \in G} \mathsf{size}(g) = nL$, does their
exist $n$ disjoint subsets $G_1, G_2, \ldots, G_n$ of $G$, each containing three elements,
such that for all $1 \leq k \leq n$, $\sum_{g \in G_k}\mathsf{size}(g) = L$?  

We map each element $g \in G$ to a unique content $c \in C$ with $\lambda_{c}=\mathsf{size}(g)$.
Thus we have $m=3n$ contents. 
The sets $G_1, G_2, \ldots, G_n$ correspond
to $n$ servers each of which can store $d_s=3$
contents and simultaneously serve $U_s=L$
requests. 
Clearly, for the above defined instance
of problem \eqref{opt_final}, the optimal objective function value is bounded
above by $nL$.
We now show that the objective function 
value exactly equals the upper bound $nL$ if and only if
there exists a solution of the 3-partition problem.

Suppose that the optimal objective function value
of the above defined instance of~\eqref{opt_final} is $nL$
and it is achieved at $Z^*=(z^*_{sc})_{s \in S, c \in C}$, i.e.,
 $nL=\sum_{s \in S c \in C}z^*_{sc}$. 
Since for each $s \in S$, $\sum_{c \in C}z_{sc}^{*} \leq L$,
and for each $c \in C$, $\sum_{s\in S} z_{sc}^{*}\leq \lambda_c$,
we must have $\sum_{c \in C} z^*_{sc}=L$,  $\forall s \in S$
and $\sum_{s \in S} z^*_{sc}=\lambda_c > 0 $, $\forall c \in C$
%
Hence, for each content $c \in C$ there must be one server $s\in S$ such that $z_{sc}^* >0$, i.e., 
$\sum_{s \in S} \indic \brac{z^*_{sc} > 0} \geq  1.$
The above implies $\sum_{c \in C} \sum_{s \in S}  \indic \brac{z^*_{sc} > 0} \geq 3n$.
%
But since every server can store at most three contents, we also have
$\sum_{c \in C} \sum_{s \in S}  \indic \brac{z^*_{sc} > 0} \leq 3n.$
Hence, we have $\sum_{c \in C} \sum_{s \in S}  \indic \brac{z^*_{sc} > 0} = 3n$,
which implies that $\sum_{s \in S} \indic \brac{z^*_{sc} > 0} = 1$, $\forall c \in C$
and $\sum_{c \in C} \indic \brac{z^*_{sc} > 0} = 3$, $\forall s \in S$.
%
Hence, a solution of the 3-partition problem is found.  
The converse follows following the same line of arguments.


\section{Proof of Theorem~\ref{thm:approx}}
\label{proof:approx}

To prove the theorem we first introduce a slightly modified version of problem~\eqref{opt_final}
by adding a constraint which requires
that only pairs $(s,c)$ belonging to a given set $\Gamma \subseteq S \times C$
can be assigned a non-zero value of $z_{sc}$, 
i.e., $z_{sc}=0$ for all $(s,c) \notin \Gamma$. 
This modified problem  is denoted as 
$\mathtt{JAM}(\mc{U},\Lambda,\mc{D}, \Gamma)$
and its optimal value is denoted as $\mathtt{JAM}^*(\mc{U},\Lambda,\mc{D}, \Gamma)$.
We note that $\mathtt{JAM}(\mc{U},\Lambda,\mc{D})$ is a special case of
$\mathtt{JAM}(\mc{U},\Lambda,\mc{D}, \Gamma)$
with $\Gamma= S \times C$. 
Clearly, if $\mc{U} \leq \mc{U}'$, $\Lambda \leq \Lambda'$, $\mc{D} \leq \mc{D}'$, $\Gamma\subseteq \Gamma'$, 
then $\mathtt{JAM}^*(\mc{U},\Lambda,\mc{D}, \Gamma) \leq \mathtt{JAM}^*(\mc{U}',\Lambda',\mc{D}', \Gamma')$.
Let $H^*=\cbrac{(s,c) \in S \times C: z_{sc}^* > 0}$ denote the collection of $(s,c)$
	pairs that are assigned non-zero values of $z_{sc}^*$
	in the optimal solution of $\mathtt{JAM}(\mc{U},\Lambda,\mc{D})$. Clearly, we have
	$\mathtt{JAM}^*(\mc{U},\Lambda,\mc{D}, H^*)=\mathtt{JAM}^*(\mc{U},\Lambda,\mc{D})$.

	We now generate a sequence of tuples $(\mc{U}_k,\Lambda_k,\mc{D}_k, H^*_k, G_k)$ 
	for each iteration $k \geq 0$ of the \texttt{greedy} algorithm (applied on $\mathtt{JAM}(\mc{U},\Lambda,\mc{D})$) as follows:
	For $k=0$ we set $\mc{U}_0=\mc{U},\Lambda_0=\Lambda,\mc{D}_0=\mc{D}, H^*_0=H^*$, $G_0=S\times C$.
	For subsequent values of $k$, let $f_k (=\mathtt{MatchedFlow})$ denote the flow
	found by the \texttt{greedy} algorithm in its $k^{\textrm{th}}$ iteration
	and let $(s_k,c_k)$ ($=(s^*,c^*)$) denote the corresponding server-content pair.
	We set $G_{k+1} = G_k-(s_k,c_k)$, $mc{U}_{k+1} = \mc{U}_{k}-f_k e_{s_k}^{(n)}$,
	$\Lambda_{k+1} = \Lambda_{k}-f_k e_{c_k}^{(m)}$, $\mc{D}_{k+1} = \mc{D}_{k}-e^{(n)}_{s_k}$,
	and $H^{*}_{k+1} = H^{*}_k-(s_k,c_k)$
	%
	where $e_r^{(l)}$ denotes
	the standard $l$-dimensional unit vector having one in the $r^{\textrm{th}}$
	component.  
	
	Note that it may so happen 
	that $(s_k,c_k) \notin H_k^*$ for some $k$ (but $(s_k,c_k)$ will always be in $G_k$). In such cases,
	the above update  sets $H_{k+1}^*=H_{k}^*$. Hence, $H^{*}_k \subseteq G_k$ is maintained for all $k$.
	Therefore, we have $\mathtt{JAM}^*(\mc{U}_k,\Lambda_k,\mc{D}_k, H^*_k) \leq \mathtt{JAM}^*(\mc{U}_k,\Lambda_k,\mc{D}_k, G_k)$.
	If the \texttt{greedy} algorithm terminates in the $q^{\textrm{th}}$
	iteration, then we must have $f_q=0$ and $\mathtt{JAM}^*(\mc{U}_q,\Lambda_q,D_q, G_q)=0$.
	Therefore, $\mathtt{JAM}^*(\mc{U}_q,\Lambda_q,\mc{D}_q, H^{*}_q)=0$. 
	Furthermore, for each iteration $k$ we prove that the following inequality
	\begin{multline}
	\mathtt{JAM}^*(\mc{U}_k,\Lambda,\mc{D}_k, H^{*}_k)\\
	-\mathtt{JAM}^*(\mc{U}_{k+1},\Lambda_{k+1},\mc{D}_{k+1}, H^{*}_{k+1}) \leq 2f_k,
	\label{eq:first}
	\end{multline} 
	%
	%
	holds. To see the above, first let us consider the case when $(s_k,c_k) \in H_k^*$. Hence,
	$(s_k,c_k) \notin H_{k+1}^*$. Let $v \geq 0$ be the flow assigned to the $(s_k,c_k)$ pair
	in the optimal solution of $\mathtt{JAM}(\mc{U}_k,\Lambda_k,\mc{D}_k, H^{*}_k)$.
	We have $\mathtt{JAM}^*(\mc{U}_k,\Lambda_k, \mc{D}_k, H^{*}_k) =v
	+\mathtt{JAM}^*(\mc{U}_{k}-ve^{(n)}_{s_k},\Lambda_{k}-ve^{(m)}_{c_k},\mc{D}_{k+1}, H^{*}_{k+1})
	\leq v+\mathtt{JAM}^*(\mc{U}_{k+1},\Lambda_{k+1},\mc{D}_{k+1}, H^{*}_{k+1})+2(f_k-v).$
	%
	The last inequality holds since no more than $2(f_k-v) \geq 0$ additional flow 
	can be matched under constraints given by the tuple 
	$(\mc{U}_{k}-ve^{(n)}_{s_k},\Lambda_{k}-ve^{(m)}_{c_k},\mc{D}_{k+1}, H^{*}_{k+1})$
	as compared to constrains given by the tuple $(\mc{U}_{k+1},\Lambda_{k+1},\mc{D}_{k+1}, H^{*}_{k+1})$.
	Inequality \eqref{eq:first} hence follows for $(s_k,c_k) \in H_k^*$.
	Now consider the case when $(s_k,c_k) \notin H_k^*$. Hence, $H_k^*=H_{k+1}^*$.
	Clearly, no more than $2f_k$ additional flow can be matched under the constraints
	given by $(\mc{U}_k, \Lambda_k, \mc{D}_k, H^{*}_k)$ as compared
	to the constraints given by $(\mc{U}_{k+1},\Lambda_{k+1},\mc{D}_{k+1}, H^{*}_{k+1})$.
	Hence, \eqref{eq:first} holds in this case also.
	Summing~\eqref{eq:first} for $k=0,1,\ldots,q-1$ we obtain $\mathtt{JAM}^*(\mc{U}_0,\Lambda_0,\mc{D}_0, H^{*}_0)-\mathtt{JAM}^*(\mc{U}_{q},\Lambda_{q},\mc{D}_{q}, H^{*}_{q})
	\leq 2\sum_{k=0}^{q-1}f_k$.
	%
	This completes the proof since 
	$\sum_{k=0}^{q-1} f_k$ is the output of 
	the \texttt{greedy} algorithm, $\mathtt{JAM}^*(\mc{U}_0,\Lambda_0,\mc{D}_0, H^{*}_0)=\mathtt{JAM}^*(\mc{U},\Lambda,\mc{D})$, 
	and $\mathtt{JAM}^*(\mc{U}_{q},\Lambda_{q},\mc{D}_{q}, H^{*}_{q})=0$. \qed 

\section{Proof of Lemma~\ref{lem:greedy}}
\label{proof:greedy}

First we note that for sufficiently large $n$, we have $n\bar{\lambda}_c > U_i$
for all $c \in C$ and all $i \in \mc{I}$. 
Hence, to allocate more than one content to the cache of a server the \texttt{greedy}
algorithm must reach a stage where the remaining flow of each
content is less than $U_{\max}:=\max_{i \in \mc{I}} U_i$. 
Since at least $U_{\min}:=\min_{i \in \mc{I}} U_i$ flow can be matched to each server, 
the maximum number of servers which can be assigned a non-zero flow after this stage is $mU_{\max}/U_{\min}$.
These are the only servers which can be allocated more than one contents.  
Therefore, the fraction of servers assigned more than two contents is at most $mU_{\max}/nU_{\min}$
which approaches zero as $n \to \infty$. This shows that the  probability of a server storing more than
one content approaches zero as $n \to \infty$, i.e., $q_{K}^{(ij)}=0$ for $\abs{K} \geq 2$.

Next, consider the case $\rho \leq 1$. 
Again, for sufficiently large $n$, we have $n\bar{\lambda}_c > U_{\max}$
for all $c \in C$.
In this case, the \texttt{greedy} algorithm cannot terminate before 
the remaining flow for each content becomes less than or equal to $U_{\max}$.  
To prove this, let us assume the converse, i.e.,
the greedy algorithm terminates when the remaining flows
for some contents are still strictly above $U_{\max}$.
This implies that the \texttt{greedy} algorithm terminated
because the remaining flows of each server has become zero.
Clearly, for this to happen we must have $n\bar{\lambda} > n \sum_{i,j} \alpha_{ij} U_i$, i.e., 
$\rho > 1$, which is a contradiction.
Therefore, the \texttt{greedy} algorithm 
terminates with 
less than $U_{\max}$ remaining flow for each content. Thus, the 
fraction of the total flow $n\bar{\lambda}$ which remains unmatched
is at most $\frac{mU}{n\bar\lambda}$, which approaches
to zero as $n \to \infty$. Hence, in the limiting system,
the whole flow $n\bar{\lambda}_c$ of each content $c \in C$ is matched.
Since $n\theta_c$ denotes the total flow of content $c$ assigned to all the servers
combined, we must have $ \theta_c= \bar{\lambda}_c$.

For $\rho > 1$, it is easy to see that the \texttt{greedy} algorithm terminates
when remaining flows of all the servers become zero.
Furthermore, at termination, all the contents,
which have been chosen at least once by the algorithm in some iteration,
have the same remaining flow. Let this flow be equal to $f$
and let the contents chosen by the algorithm at least once be $c=1,2,\ldots,c^*$
for some $c^* \leq m$. Then, we must have $n \bar{\lambda}_{c^*+1} \leq f < n \bar{\lambda}_{c^*}$.
Also, since the total matched flow $n\sum_{c'=1}^{c^*}\bar{\lambda}_{c'}-kf$ combining all the contents
is equal to the total capacity $\bar{\lambda}/\rho$ of the system, we have
$f=\frac{n}{c^*}\brac{\sum_{c'=1}^{c^*} \bar\lambda_{c'}-\frac{\bar{\lambda}}{\rho}}$.
The matched flow for each content $c=1,2,\ldots,c^*$ is therefore $n\theta_c=n\bar{\lambda}_c-f$
and for each content $c >c^*$ is $\theta_c=0$. This completes the proof of the lemma. \qed

\section{Proof of Theorem~\ref{thm:mfconv}}
\label{proof:mfconv}

We recall from Theorem~1 of~\cite{Gast_SIG} 
that the DI $\dot{x} \in H(x)$ has at least one solution $x$
with $x(0)=x_0 \in \mc{W}$ if 1) for each $w \in \mc{W}$ the set 
$H(w)$ is non-empty, closed, convex; 
2)  $\norm{H(w)}:={\sup\cbrac{\norm{z}_2: z\in H(w)}} < D(1+\norm{w})$ 
for some constant $D >0$; 3) $H$ is upper semi-continuous\footnote{The set valued mapping
	$F$ is said to be upper-hemicontinuous at $w$ if $\forall w_n, \forall z_n \in F(w_n)$,
	$\lim_{n\to \infty} w_n=w$ and $\lim_{n \to \infty}z_n=z$ implies $z \in H(w)$.}.
Furthermore, the density dependent Markov process $x^{(n)}$ with limiting drift
$h$ converges in probability to the solution of the DI $\dot{x} \in H(x)$ if 
$H(w)=\mathsf{conv}\brac{\mathsf{acc}_{w_k \to w} h(w_k)}$, where
$\mathsf{conv}(V)$ denotes the closure of convex hull containing the set $V$
and $(\mathsf{acc}_{w_k \to w} h(w_k))$ denotes the set of accumulation
points of the sequence $(h(w_k))$ for $w_k \to w$.  

From \eqref{eq:di1}, it follows directly that $H(w)$ is nonempty, closed and convex for each $w \in \mc{W}$.
Furthermore, for each $(i,j,K) \in \mc{L}$ and $r \in \mb{Z}_+$ we have 
$0 \leq \bar{\lambda}_{c}\frac{(w_{i,j,K,r-1}-w_{i,j,K,r})}{\sum_{K': c \in K'} \sum_{i,j} \alpha_{ij}q_{ijK'}(1-w_{i,j,K',U_i})} 
\leq  \bar\lambda_{c}/ \alpha_{ij}q_{ijK}$. Hence, from \eqref{eq:di1} we have
\begin{equation*}
{H_{i,j,K,r}(w)} \leq D_{i,j,K}:=\sum_{c \in K} \frac{\bar\lambda_{c}}{\alpha_{ij}q_{ijK}}+ \max_{i \in \mc{I}}U_{i}
\end{equation*}
Therefore, $\norm{H(w)} \leq D \leq D(1+\norm{w}_2)$, where $D:=\sqrt{\sum_{(i,j,K) \in \mc{L}} (1+U_i D_{i,j,K}^2)} >0$. 
We also note that $H_{i,j,K,r}(w)$ is continuous if $w_{i,j,K,U_i} < 1$ and 
the compact set $[0, \bar\lambda_{c}/ \alpha_{ij}q_{ijK}]$ contains all limit points
of 	$\bar{\lambda}_{c}\frac{(w_{i,j,K,r-1}-w_{i,j,K,r})}{\sum_{K': c \in K'} \sum_{i,j} \alpha_{ij}q_{ijK'}(1-w_{i,j,K',U_i})}$
as $w_{i,j,K,U_i} \to 1$. Hence, the set valued mapping $H$ is upper-hemicontinuous. 
 By definition it follows that  $H(w)=\mathsf{conv}\brac{\mathsf{acc}_{w_k \to w} h(w_k)}$.
Therefore, the statement of the theorem follows from 
Theorem~1 of~\cite{Gast_SIG}. \qed

\section{Proof of Theorem~\ref{thm:global_stability}}
\label{proof:global_stability}

Let $y_{c}(t):=\sum_{K: c \in K}\sum_{i,j}\alpha_{ij} q_{ijK}\sum_{r=1}^{U_i} x_{i,j,K,r}(t)$
for all $t \geq 0$.
For the \texttt{greedy} algorithm, we have from Lemma~\ref{lem:greedy} that $q_{ijK}=0$
for $\abs{K} > 1$. Hence,  we have $y(t)=\sum_{c \in C} y_c(t)$. Further,
in \eqref{eq:di1} substituting $K=\cbrac{c}$ and 
summing we obtain 
%
\begin{equation}
\dot{y}_c(t)\in H_c(y_c):= [0,\bar\lambda_c]\indic\brac{y_c = \theta_c}
+\bar\lambda_c\indic\brac{y_c < \theta_c}-y_c.
\label{eq:di2}
\end{equation}
It can be easily verified that $(z_1-z_2)(w_1-w_2) \leq 0$ for all $w_1,w_2 \in \mb{R}$ and for all $z_1 \in H_c(w_1), z_2 \in H_c(w_2)$.
In other words, $H_c$ is one-sided Lipschitz (OSL) with Lipschitz constant $L=0$.
Therefore, the DI~\eqref{eq:di2} has a unique solution. It can also be verified that for any $y(0)=y_0 \in [0,\theta_c]$,
$y_c(t)=\tilde{y}_c(t)\indic\brac{\tilde{y}_c(t) < \theta_{c}}+
\theta_{c}\indic\brac{\tilde{y}_c(t) \geq \theta_{c}}$, where $\tilde{y}_c(t)=\bar\lambda_c+(y_0-\bar\lambda_c)e^{-t}$
is a solution of the DI~\eqref{eq:di2}. Hence, it must be the only solution.
Since from Lemma~\ref{lem:greedy} we have for $\rho\leq 1$, $\theta_c=\bar\lambda_c$ for all $c \in C$,
it follows that $y_c(\infty)=\lim_{t\to \infty} y_c(t)=\bar\lambda_c$, or $y(\infty)=\bar\lambda$.
For $\rho > 1$ again from Lemma~\ref{lem:greedy} we have that $\theta_c < \bar\lambda_c$ for all
$c=1:c^*$ and $\theta_c=0$ for $c=c^*+1:m$.  Hence, $y_c(\infty)=\theta_c$ for $c=1:c^*$
and $y_c(\infty)=0$ for  $c=c^*+1:m$. Thus, $y(\infty)=\sum_{c \in C}y_c(\infty)=\sum_{c=1}^{c^*} \theta_c=\bar\lambda/\rho$.

Now, we consider the  \texttt{p2p} algorithm with $d_s=1$ for all $s \in S$.
For this case we have $q_{ijc}=\hat{\lambda}_c$ for all $c \in C$
and $q_{ijK}=0$ if $\abs{K} > 1$.  Hence, as before we have $y(t)=\sum_{c \in C} y_c(t)$
and  $y_c(t)=\tilde{y}_c(t)\indic\brac{\tilde{y}_c(t) < \theta_{c}}+
\theta_{c}\indic\brac{\tilde{y}_c(t) \geq \theta_{c}}$. 
In this case, 
$\theta_c=\bar\lambda_{c}/\rho$.
We thus have, for $\rho \leq 1$, $\bar\lambda_c \leq \theta_c$ and for $\rho \leq 1$
$\theta_c < \lambda_{c}$.
Hence, $y_c(\infty)=\lim_{t \to \infty}y_c(t)=\lambda_c$ for $\rho \leq 1$
and $y_c(\infty)=\theta_c=\lambda_{c}/\rho$ for $\rho > 1$. Thus,
$y(\infty)=\min\brac{\bar{\lambda},  \bar\lambda/\rho}$.

Since $y^{(n)}(\infty) \leq { \bar\lambda/\rho}$ uniformly for all $n$,
it follows that the sequence $(y^{(n)}(\infty))_n$ is tight. From Theorem~\ref{thm:mfconv} 
and the uniqueness of solution of $y(t)$ we have that $y^{(n)} \to y$.
Hence, every limit point of the sequence of stationary measures of 
$y^{(n)}$ must be an invariant measure of the process $y$.  Since $y(\infty)$
is the unique, globally asymptotically stable stationary point of the process 
$y$ it follows that the only invariant
measure for the process $y$ is the Dirac measure concentrated at $y(\infty)$.
Thus, all limit points of the sequence of stationary measures of 
$y^{(n)}$ must coincide with the Dirac measure at $y(\infty)$, i.e., $\lim_{n\to \infty}y^{(n)}(\infty)=y(\infty)$. \qed

\end{document}